\documentclass[conference]{IEEEtran}
\IEEEoverridecommandlockouts
\usepackage{cite}
\usepackage{epstopdf}
\usepackage{epsf,pgf,xcolor}
\usepackage{epstopdf,epsfig}
\usepackage[utf8]{inputenc}
\usepackage{epsf}
\usepackage{float}
\usepackage{stfloats}
\usepackage{amsthm}
\usepackage[english]{babel}
\usepackage{graphicx}
\usepackage{algorithm2e}
\usepackage{amsmath,amssymb,amsfonts}
\usepackage{algorithmic}
\usepackage{graphicx}
\usepackage{textcomp}
\newtheorem{remark}{Remark}
\newtheorem{theorem}{Theorem}
\usepackage{algorithm2e}
\usepackage{graphicx}
\usepackage{dcolumn}
\usepackage{bm}
\usepackage[export]{adjustbox}
\usepackage{array}
\usepackage{caption}
\usepackage{commath}
\usepackage[T1]{fontenc}
\usepackage{tikz}
\usetikzlibrary{matrix,shapes,arrows,positioning,chains}
\usepackage{url}
\usepackage{ifthen}
\usepackage{caption}
\usepackage{mathtools}
\DeclarePairedDelimiter\ceil{\lceil}{\rceil}
\DeclarePairedDelimiter\floor{\lfloor}{\rfloor}
\usepackage{amssymb}
\newtheorem{lemma}{Lemma}
\newtheorem{proposition}{Proposition}

\newcommand\bl[1]{{\color{blue}#1}}
\graphicspath{ {./main/} } 

\usepackage{xcolor}
\def\BibTeX{{\rm B\kern-.05em{\sc i\kern-.025em b}\kern-.08em
    T\kern-.1667em\lower.7ex\hbox{E}\kern-.125emX}}
 
 \usepackage{verbatim}
   
\begin{document}

\title{On the Age of Information in Erasure Channels with Feedback
}

\author{%
   \IEEEauthorblockN{Alireza Javani,
                     Marwen Zorgui,
                     and Zhiying Wang
                     }
   \IEEEauthorblockA{
                     Center for Pervasive Communications and Computing\\
  					University of California, Irvine\\
                     \{ajavani, mzorgui, zhiying\}@uci.edu}

 }

\maketitle

\begin{abstract}
We consider a status updating system where having timely knowledge about the information source at the destination (monitor) is of utmost importance. By utilizing the age of information (AoI) metric,  the freshness of status update over an erasure channel is investigated. Due to the erasure nature of the update transmission, an error-free feedback channel from the monitor to the source is necessary for reducing AoI. Each status update contains $K$ packets which can be sent through the channel one at a time. At each channel use, one status update is available from the information source. Depending on how many packets have been received successfully out of the $K$ packets, we need to decide whether to continue sending the current update or terminate it and start sending the newly generated update. In this paper, we find the optimal failure tolerance when the erasure probability ($\epsilon$) is in the regime $\epsilon \to 0$ and also provide a lower and an upper bound for the average AoI for all erasure probabilities. Moreover, for all $\epsilon$, we provide a lower bound for failure tolerance to minimize peak AoI. 

\end{abstract}

\vspace{-0.25cm}
\section{Introduction}
Timeliness has become a major requirement in many applications such as wireless sensor networks, environmental monitoring, health monitoring, as well as applications in vehicular networks, surveillance systems and internet of things (IoT). The communication channel between the source of information and the desired destination is required to transmit observations from the status of the information source in a timely manner.  

In \cite{kaul2012real}, age of information (AoI) is introduced to quantitatively measure the timeliness of the transmitted information. 
Updates are generated at the information source, and delivered to the destination (monitor).
AoI is defined as the amount of time elapsed since the generation of the last received update and average AoI is the average age over all time. Extensions to
networks of multiple sources and servers with and without
packet management  are studied in \cite{yates2018age,javani2019age,abdelmagid2019reinforcement,kam2014effect,costa2014age,abdelmagid2019deep}.
In order to quantify the maximum age at the monitor, peak AoI is introduced in \cite{costa2016age}  as the age at the monitor right before receiving an update.

The focus of this paper is to study age of information in an erasure channel with feedback. Several previous works addressed scenarios under erasure channels and/or feedback.
In \cite{najm2017status}, the authors consider a system where updates are generated according to a Poisson process and sent through an erasure channel. Two hybrid ARQ protocols are considered: infinite incremental redundancy (IIR), and fixed redundancy (FR). 
In \cite{yates2017timely}, the authors assume a just-in-time generation process and transmission over an erasure channel. They investigate the IIR and FR schemes. Authors in \cite{najm2019optimal} also consider transmission over an erasure channel where the generation of source updates is assumed to be periodic. 
Recently in \cite{chen2019benefits, feng2019adaptive}, AoI is studied in the setting of two-user broadcast symbol erasure channels with feedback. In particular,  in \cite{chen2019benefits}, the benefits of network coding in terms of age are investigated, while in \cite{feng2019adaptive}, the authors propose an adaptive coding scheme achieving small AoI at both users. 

In this paper, we consider a system model similar to \cite{najm2019optimal} but with feedback.
In particular, updates of a source are transmitted through a symbol erasure channel to a monitor. Each source update is comprised of $K$ \emph{channel symbols}, also called \emph{packets}, and a new source update is available per channel use (or time unit). The monitor employs an error-free feedback channel to notify the source of symbol erasures. We note that the feedback cost per packet is $1$ bit, which is negligible if the packet size is large. From this perspective, having a feedback channel does not incur a big cost on the system while helping reduce the age of information.

We ask the following question: Given the knowledge of the previous successful and erased packet transmissions for an update, should the source continue to transmit or drop the update, in order to minimize the age of information? 
We first derive an expression of average AoI related to the number of time units for the terminated update transmissions and for the successful update transmission.
Based on that, policies with zero and infinite error tolerance are investigated.
When $\epsilon$ is close to $0$, the optimal policy is proved to be zero error tolerance for the first two packets. 
We also provide a lower and an upper bound for optimal average AoI. It is observed that the upper bound based on the infinite error tolerance policy is numerically close to the optimal policy. 

Moreover, we investigate the average peak AoI. We prove that the error tolerance should increase as the number of successful packets increases. We also provide a bound for the optimal policy parameters, and simulation shows that its peak AoI is close to optimal.

The paper is organized as follows. In Section~\ref{sec:pre}, we introduce our system model and present general expressions for calculating the average and the peak AoI. The average AoI is studied in Section~\ref{sec:ave} under several scenarios.
The peak AoI is investigated for several special cases in Section~\ref{sec:peak}, and Section~\ref{sec:conclude} concludes the paper.

\noindent \textbf{Notation}. For a non-negative integer $n$, define $[n]~\coloneqq ~\{1, \ldots, n\}$.
\vspace{-0.2cm}
\section{System Model and Preliminaries}\label{sec:pre}
In this section, we present our model. The model consists of an information source, sending its information over an erasure channel. Each source symbol (update) is composed of $K$ channel symbols (packets), where each of these packets takes one channel use, and can be erased with probability~$\epsilon$.
Also, we assume that upon each channel use, a new update is available. At each channel use we have the option of continuing to send the remaining packets of the current update being transmitted, or terminating the current update transmission and start sending the newest available update. We assume the existence of an error-free feedback channel from the monitor to the source indicating whether a packet has been received successfully. 

The goal is to find a policy that minimizes AoI. 
AoI is defined as the time duration from the generation of an update at the information source to the current time.
Formally, the average AoI is defined as \cite{kaul2012real}
\vspace{-0.2cm}
\begin{align} \label{graph}
\Delta = \lim_{T \to \infty} \frac{1}{T} \int_0^{T} \Delta(t) dt.
\end{align}
Here $\Delta(t) = t - u(t)$ and $u(t)$ is the generation time of the most recent update at the monitor.
The (average) peak AoI $(PAoI)$ is defined as the value of age right before arrival of an update at the destination, averaged over all received updates.

To illustrate the challenges in minimizing AoI with feedback, we consider the following two cases. Consider that the source sends the first packet of an update and it is erased. Recall that per our system model, a source update is available for each channel use. Thus, it is obvious that the source should drop the current update being transmitted and send the first packet of the newly generated update in order to minimize AoI. If the first packet of the current update is successfully received, then the source should continue with the transmission of the current update. On the other hand, assume that the source has transmitted successfully $(K-1)$ packets out of~$K$~packets in $(K-1)$~channel uses, where $K$ is very large, and an erasure happens during the transmission of the last packet. Then, intuitively, the source should try sending the $K$-th packet at least one more time, instead of dropping the current update and starting over with a fresher one. 

 
In general, given~$K$ and~$\epsilon$, the decision of whether to continue transmission of the current update depends on the total number of successfully delivered packets at the considered instant, and also on the total time elapsed since the start of its transmission. In this work, we focus on a particular family of policies, described below.  
%
For $i \in \{1,...,K\}$, we define $c_i$ such that if $c_i$ consecutive erasures happen during the transmission of the $i$-th packet, the source drops the current update and starts over with the newly generated source update. That is,  $c_i -1$ is the maximum number of allowed erasures while transmitting the $i$-th packet of an update. A \emph{policy} is described by a vector $\mathbf{c}=[c_1, c_2, \ldots, c_K]$. Let $\Delta(\mathbf{c}), PAoI(\mathbf{c})$ denote the average AoI and the average peak AoI under the policy $\mathbf{c}$.
The main problem of interest in this paper is the following:
\begin{equation}
\begin{aligned}
& \underset{\mathbf{c}=[c_1, c_2, \ldots, c_K] }{\text{minimize}}
& & \Delta(\mathbf{c}) \text{ or } PAoI(\mathbf{c}), \\
& \text{subject to}
& &   c_i \in \mathbb{N} \backslash \{0\},  \; i \in \{1,...,K\}.
\end{aligned}
\end{equation}
Here $\mathbb{N}\backslash \{0\}$ denotes the set of positive integers. When it is clear from the context, we drop the argument and write $\Delta$ and $PAoI$. A policy depends on the number of successfully transmitted packets of an update, but not on the total number of transmissions.

We take a similar approach as \cite{yates2018age} for deriving the AoI formula in our setting.
As illustrated in Figure~\ref{AoI}, each $S_i$ indicates the time for the $i$-th successful update transmission, at which AoI jumps to a lower value as the destination has a fresher update from the information source. Because the channel is an erasure channel, one or multiple failed attempts to send an update to the destination may occur between two successful updates.
Each $D_i$ represents the total time for failed updates between the $(i-1)$-th and the $i$-th successful update. In particular, AoI increases linearly in the duration of $D_i+S_i$. We are interested in average AoI, or  the area under the curve per time unit, and peak AoI, or the peak points on the curve.
\begin{figure}
    \centering
\hbox{\hspace{-0.5cm}    \includegraphics[scale=0.38]{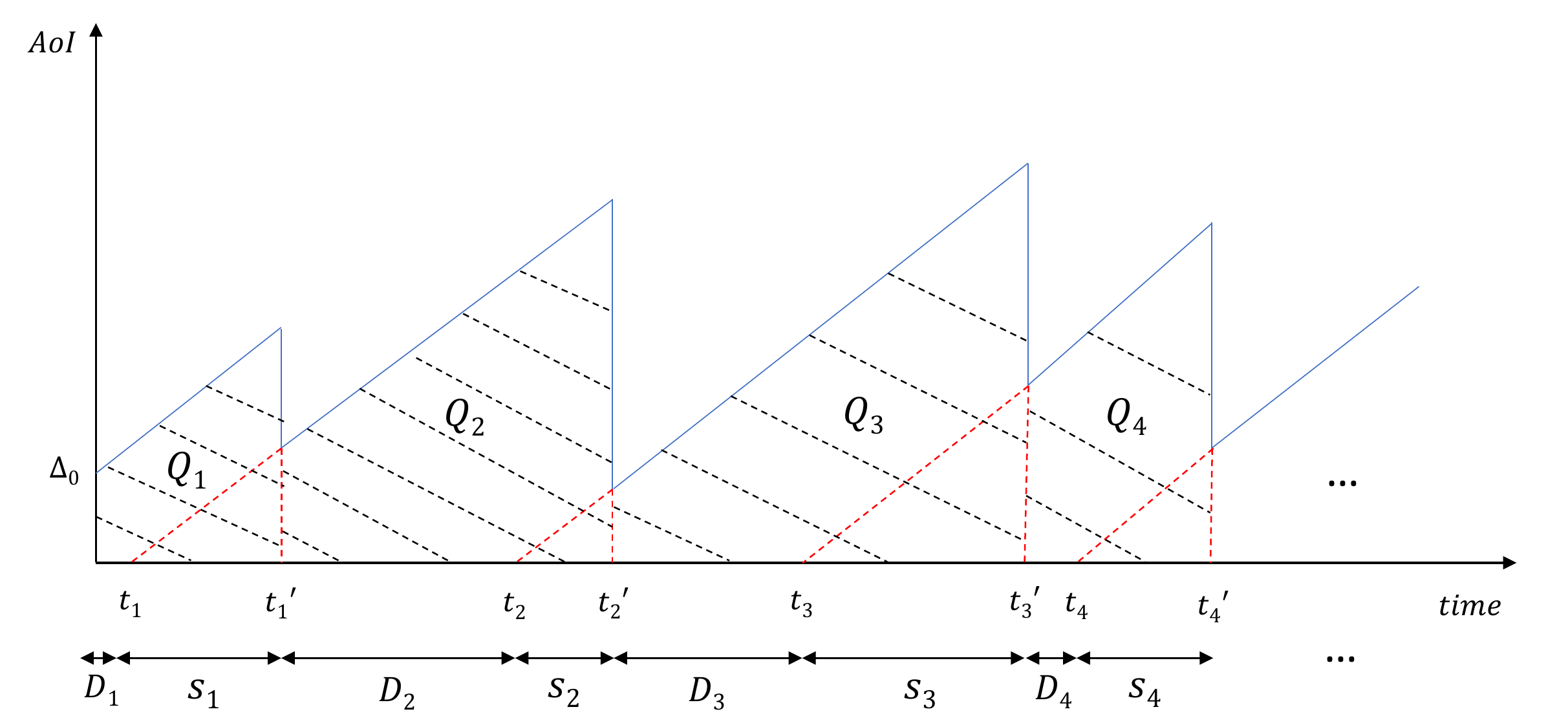} }
    \caption{AoI for the erasure channel with feedback. $t_i$ represents the generation time of the $i$-th successful update, and $t_i'$ denotes its delivery time.}
    \label{AoI}
\end{figure}
\begin{lemma}\label{lem1}
 Let $S, D$ be the random variables representing the time duration for a successful update transmission, and for all the failed update transmissions between successful ones, respectively. The average AoI is
\begin{align}  \label{formual}
\Delta =\mathbb{E}[S] + \frac{\frac{1}{2} \mathbb{E}[S^2] +\frac{1}{2} \mathbb{E}[D^2] + \mathbb{E}[S]\mathbb{E}[D]}{\mathbb{E}[S]+\mathbb{E}[D]}.
\end{align}
The (average) peak AoI is 
\begin{align} \label{formula_P}
    PAoI = 2\mathbb{E}[S]+\mathbb{E}[D].
\end{align}
\begin{proof}
Let $N(T) \coloneqq \max\{ {i | {t_{i}}^{\prime} \leq T}\}$ where ${t_i}^{\prime}$ is arrival time of the $i$-th successful update. The area under the AoI graph consists of several distinct trapezoid areas called $Q_i$ as shown in Figure \ref{AoI}, from time $t'_{i-1}$ to $t'_i$. Let $\Tilde{Q}$ represent the finite residual area in the integral when $N(T) < t \leq T$. Using the definition in \eqref{graph}, we write 
\vspace{-0.05cm}
\begin{small}
\begin{align}\label{q}
\frac{1}{T} \int_0^{T} \Delta(t) dt =  \frac{Q_1+ \Tilde{Q}}{T} + \frac{N(T)-1}{T} \frac{\sum_{i=2}^n Q_i}{N(T)-1}.
\end{align}
\end{small}
Each $Q_{i}$ can be calculated by subtracting $2$ adjacent isosceles. 
Consequently, we have $Q_i= \frac{1}{2} (S_i + D_i+ S_{i-1})^2 - \frac{1}{2} S_{i-1}^2= \frac{1}{2} (S_{i}^2+D_{i}^2+2S_i D_i + 2 S_i S_{i-1}+ 2D_i S_{i-1})$. The first term in \eqref{q} corresponds to the boundary effect and is negligible when $T \to \infty$. Since the channel and the policy do not change over time, the system is stationary and ergodic.  As $T \to \infty$. 
\vspace{-0.2cm}
\begin{equation*}
 \frac{Q_1+\Tilde{Q}}{T} \to 0, \text{ } \frac{N(T)-1}{T} \to \frac{1}{\mathbb{E}[S+D]}, \text{ } \frac{\sum_{i=2}^n Q_i}{n-1}= \mathbb{E}[Q].
\end{equation*}
Since the random variables $S$ and $D$ are independent in our setting, \eqref{graph} reduces to 
\begin{align*}
\Delta =& \frac{\frac{1}{2} (\mathbb{E}[S^2]+\mathbb{E}[D^2]+2\mathbb{E}[S]\mathbb{E}[D]+2\mathbb{E}[S]^2+2\mathbb{E}[D]\mathbb{E}[S])}{\mathbb{E}[S+D]} \nonumber\\
= &\mathbb{E}[S] + \frac{\frac{1}{2} \mathbb{E}[S^2] +\frac{1}{2} \mathbb{E}[D^2] + \mathbb{E}[S]\mathbb{E}[D]}{\mathbb{E}[S]+\mathbb{E}[D]}.
\end{align*}
The peak AoI, as from Figure~\ref{AoI}, is given by $S_{i-1} + S_{i}+D_{i}$, which is the value of age right before the arrival of $i$-th successful update at the destination. Therefore, the average peak AoI equals to $\mathbb{E}[S]+ \mathbb{E}[S+D]= 2\mathbb{E}[S]+ \mathbb{E}[D]$. 
\end{proof}
\end{lemma}
\vspace{-0.2cm}
\begin{remark}
When the communication channel is erasure-free, $\mathbb{E}[S]=K,\mathbb{E}[S^2]=K^2 $, and $\mathbb{E}[D]=\mathbb{E}[D^2]=0$. Therefore, by substituting these values in \eqref{formual} we achieve $\Delta = \frac{3K}{2}$. 
\end{remark}
\begin{lemma} \label{lem2}
For a given policy corresponding to a vector $\mathbf{c}$, the quantities in \eqref{formual} are given by
\begin{align}
\label{expect_S}
\mathbb{E}[S]&= \frac{(1-\epsilon)^K}{p}  \left(\sum_{i_1=0}^{c_1-1}   \dots \sum_{i_K=0} ^{c_K-1}  (\sum_{j=1}^{K} i_j+K) \epsilon^{\sum_{j=1}^{K} i_j} \right),\\
\label{expect_S_2}
\mathbb{E}[S^2]  &= \frac{(1-\epsilon)^K}{p}  \left(\sum_{i_1=0}^{c_1-1}   \dots \sum_{i_K=0} ^{c_K-1}  (\sum_{j=1}^{K} i_j+K)^2 \epsilon^{\sum_{j=1}^{K} i_j}\right),\\
 \mathbb{E}[D] &= \frac{1-p}{p} \mathbb{E}[d], \quad
 \mathbb{E}[D^2] = \frac{1-p}{p} \mathbb{E}[d^2]+ 2  \mathbb{E}[D]^2,
\end{align}
where $p$ is the probability of a successful delivery of an update, and $d$ is the random variable corresponding to the number of channel uses of a failed update, such that
\begin{align}
p&= (1-\epsilon^{c_1})(1-\epsilon^{c_2})\dots (1-\epsilon^{c_K}),\label{successful_update} \\
\mathbb{E}[d]&= \frac{1}{1-p} 
\left(
\sum_{j=1}^{K} \sum_{i_1= 0}^{c_1-1} \dots \sum_{i_{j-1}= 0}^{c_{j-1} -1}
(1 - \epsilon)^{j-1} \epsilon^{c_j +  \sum_{h= 1}^{j-1} i_h} 
\right.
\nonumber\\
& \left.
   \times (  j -1 + c_j +\sum_{h= 1}^{j-1} i_h )
\right)
\label{expect_d}
,\\
\mathbb{E}[d^2]&= \frac{1}{1-p}
\left(
\sum_{j=1}^{K} \sum_{i_1= 0}^{c_1-1} \dots \sum_{i_{j-1}= 0}^{c_{j-1} -1}
(1 - \epsilon)^{j-1} \epsilon^{c_j +  \sum_{h= 1}^{j-1} i_h} 
\right.
\nonumber\\
&  \left. \times (  j -1 + c_j +\sum_{h= 1}^{j-1} i_h )^2
\right).
\label{expect_d_2}
\end{align}
\end{lemma}
\begin{proof}
To calculate $\mathbb{E}[S]$, recall that a successful transmission of an update requires no termination and therefore, the number of erasures that may happen during the transmission of the $j$-th packet, denoted by $i_j$, is at most $c_{j}-1$. The probability of a successful update transmitting exactly $i_j+1$ times for packet $j$, $j \in [K]$, is given by $(1-\epsilon)^K \epsilon^{\sum_{i_{j}=1}^{K} i_j}$.
Therefore, the probability $p$ that an update is successful follows as 
\begin{align*}
    p&=(1-\epsilon)^K \sum_{i_1=0}^{c_1-1} \sum_{i_2=0} ^{c_2-1} \dots \sum_{i_K=0}^{c_K-1} \epsilon^{\sum_{i_{j}=1}^{K} i_j}\\
    &=(1-\epsilon^{c_1})(1-\epsilon^{c_2})\dots (1-\epsilon^{c_K}). 
\end{align*}
As $S$ corresponds to the amount of time to deliver a successful update, it follows that 
$\text{Pr}\left(S = K +\sum_{j=1}^{K} i_j\right)=~\frac{(1-\epsilon)^K \epsilon^{\sum_{j=1}^{K} i_j}}{p}$, and \eqref{expect_S} and \eqref{expect_S_2} hold. 
Random variable $D$ corresponds to all terminated updates between two successful ones. We write $D=d_1+d_2+...+d_M$,  where $M$ is the random varaible for the number of terminated updates, and $d_j, j \in [M],$ is the number of channel uses of the $j$-th failed update. Note that $d_j$'s are i.i.d random variables, and $M$ is independent of them. It follows that
\begin{align*}
\mathbb{E}[D] &=\mathbb{E}_M[\mathbb{E}[D|M]] = \mathbb{E}_M[\sum_{j=1}^{M} \mathbb{E}[d_j]] \\
&=\mathbb{E}_M[M \mathbb{E}[d]]= \mathbb{E}[M]\mathbb{E}[d]= \frac{1-p}{p} \mathbb{E}[d],
\end{align*}
where the random variable $d$ is the number of channel uses of a terminated update. $M$ has a geometric distribution with probability mass function $\text{Pr}(M=m)= (1-p)^m p$, for $m \geq 0$ and $p$ given by \eqref{successful_update} .
\begin{align*}
   &\mathbb{E}[D^2] = \mathbb{E}_M[(\sum_{j=1}^{M} d_j)^2] = \mathbb{E}_M[M\mathbb{E}[d^2]+ M(M-1) \mathbb{E}[d]^2 ]\\
   &=\frac{1-p}{p} \mathbb{E}[d^2]+ \frac{2(1-p)^2}{p^2} \mathbb{E}[d]^2= \frac{1-p}{p} \mathbb{E}[d^2]+ 2\mathbb{E}[D]^2.
    \end{align*}
Recall that an update is dropped whenever $c_{j}$ erasures happen during the transmission of its $j$-th packet, for any $j\in [K]$. Note that for a given failed update, the probability of the update being terminated during the transmission of its $j$-th packet is given by $\frac{(1-\epsilon)^{j-1} \epsilon^{c_{j}} ( \epsilon^{i_1+i_2+...+i_{j-1}})}{1-p}$ for some $i_h,  h \in [j-1]$, such that $i_h < c_{h}$. The number of channel uses of the considered failed update is  $i_1+i_2+...+i_{j-1} +  c_{j}+ j-1$. 
Consequently, summing over $\{ c_j, i_1, \ldots, i_{j-1} \}$, for $ j \in [K]$, we obtain $\mathbb{E}[d]$ and $\mathbb{E}[d^2]$ as in \eqref{expect_d} and \eqref{expect_d_2}.
\end{proof}
\vspace{-0.1cm}
\section{Average AoI} \label{sec:ave}
In this section, combining results from Lemma \ref{lem1} and Lemma \ref{lem2}, we study the average AoI under several scenarios.
First, we investigate the following two policies:  1) we do not tolerate any erasure during the transmission of a source update, that is, $c_i = 1 $ for $i \in[K]$ and 2) we keep transmitting the current update until all $K$ packets are successfully received, that is, $c_i = \infty $ for $i \in [K]$. We note that the second one is equivalent to the infinite incremental redundancy policy in~\cite{najm2017status,yates2017timely}. Then, we consider the general policy in the regime where the erasure probability is very small. At last, we derive upper and lower bounds of average AoI for any $\epsilon$.
\subsection{AoI with Zero Error Toleration Policy}
The zero error policy corresponding to $c= [1,\ldots,1]$ dictates that whenever an erasure happens during the transmission of the current update, the source drops it and starts transmission of the newest available update. 
\begin{theorem} Under zero error policy, average AoI is given by
\begin{small}
\begin{equation}
\Delta =\frac{3K}{2} + \frac{\frac{1}{(1-\epsilon)^K}-(1+K\epsilon)}{2\epsilon}+\frac{1}{2\epsilon} \frac{\frac{1}{(1-\epsilon)^{2K}}+ \frac{2K\epsilon+\epsilon}{(1-\epsilon)^K}+\epsilon-1}{\frac{1}{(1-\epsilon)^K}-1}.
\label{zero_error_policy}
\end{equation}
\end{small}
\end{theorem}
This result is obtained by evaluating the AoI formula~\eqref{formual} for $\mathbf{c}=[1,\ldots,1]$.

\begin{proof}
Setting $c_i = 1$, we obtain $p = (1-\epsilon)^K$,
\begin{align*}
\mathbb{E}[S] &= \frac{(1-\epsilon)^K}{(1-\epsilon)^K} K = K, \quad
 \mathbb{E}[S^2] = \frac{(1-\epsilon)^K}{(1-\epsilon)^K} K^2= K^2,\\
\mathbb{E}[D] 
&=\frac{\frac{1}{(1-\epsilon)^K}-(1+K\epsilon)}{\epsilon},\\
\mathbb{E}[D^2] &=\frac{1}{(1-\epsilon)^K}  \sum\limits_{j=0}^K (1 - \epsilon)^{j-1} \epsilon j^2 + 2\mathbb{E}[D]^2 \\
&=\frac{1}{\epsilon^2} \left(\frac{-2-\epsilon-4K\epsilon}{(1-\epsilon)^K} 
+\frac{2}{(1-\epsilon)^{2K}}
\right. \\
& \left. + (1+\epsilon K)^2 -1+\epsilon \right).
\end{align*}
Substituting the above values in the AoI formula~\eqref{formual}, and after simplifications, one obtains~\eqref{zero_error_policy}.
\end{proof}
\vspace{-0.1cm}
\subsection{AoI with Infinite Error Tolerance}
When $\mathbf{c}=[\infty,\infty, \ldots,\infty]$,  no update is terminated, irrespective of the number erasures during its transmission.
\begin{theorem}
Under infinite error policy, the average AoI is 
\begin{align}
\label{IIReq}
    \Delta = \frac{3K}{2} + \frac{\epsilon(3K+1)}{2(1-\epsilon)}.
\end{align}
\end{theorem}
\begin{proof}
Setting $c_i = \infty$, one obtains
\vspace{-0.2cm}
\begin{align*}
 p &=1, \quad  \mathbb{E}[D]=0,  \quad \mathbb{E}[D^2]=0, \\
\mathbb{E}[S]&= \frac{(1-\epsilon)^K}{p}  \Bigg(\sum_{i_1=0}^{\infty}   \dots \sum_{i_K=0} ^{\infty}  (\sum_{j=1}^{K} i_j+K) \epsilon^{\sum_{j=1}^{K} i_j}\Bigg)\\
&= K + (1-\epsilon)^K \frac{K \epsilon}{(1-\epsilon)^{K+1}}= \frac{K}{1-\epsilon}, \\
\mathbb{E}[S^2] &= \frac{(1-\epsilon)^K}{p}  \Bigg(\sum_{i_1=0}^{\infty}   \dots \sum_{i_K=0} ^{\infty}  (\sum_{j=1}^{K} i_j+K)^2 \epsilon^{\sum_{j=1}^{K} i_j}\Bigg) \\
 &=K^2 + 2K \frac{K\epsilon}{1-\epsilon} +  K \epsilon \frac{1+K\epsilon}{(1-\epsilon)^2}.
\end{align*}
Therefore using the formula in \eqref{formual} and some algebraic simplifications one obtains~\eqref{IIReq}.
\end{proof}
\subsection{AoI in the Small Erasure Probability Regime ($\epsilon \to 0$)}
Next, we consider a general policy for channels with small erasure probability. We assume the number of packets per update, $K$, is a constant and independent of $\epsilon$.
 \begin{theorem} \label{epsil0}
When the erasure probability $\epsilon \to 0$, for $K \geq 3$, the optimal average AoI can be achieved with ${c_1}^*={c_2}^*=1$, and any ${c_j}^* > 1, $ for $j \in \{3,4,...,K\}$. Moreover, we have
\begin{align}
    \Delta^{*}= \frac{3K}{2} + \frac{\epsilon}{2K} (3K^2 -2K+3) + o(\epsilon).
    \label{Aoi_low_snr}
\end{align}
\end{theorem}
\begin{IEEEproof}
When $\epsilon \to 0$, we use a first order approximation and neglect powers of $\epsilon$ greater than $1$.
%
Let  $A~=~\{x_1,x_2,...,x_l \} $ denote the set of all distinct indices in $\{1,\ldots,K\}$ such that $c_{x_i}>1$, for  $i\in \{1,\ldots,l \}$.
The remaining indices in $\{1,\ldots,K \}$ and not in $A $ are denoted by $ \{y_1,y_2,\dots,y_{K-l}\}$.
That is, $c_{y_{i}}=1$ for $i \in \{1,\ldots, K-l \}$. Following some algebraic simplifications, one obtains
\begin{align}
p&\approx  1- \sum_{i=1}^{K} \epsilon^{c_i} \approx 1- \sum_{i=1}^{K-l} \epsilon^{c_{y_{i}}} = 1- \epsilon (K-l),\\
\mathbb{E}[S]&\approx (K+ \epsilon l (K+1)) (1-\epsilon l) = K+ \epsilon l,\\ \mathbb{E}[D]&\approx \epsilon 
(\sum_{i=1}^{K-l} y_i)\nonumber\\
\mathbb{E}[S^2]&\approx (K^2+ \epsilon l (K+1)^2) ((1-\epsilon l))\approx K^2 + \epsilon l (2K+1)\nonumber\\
\mathbb{E}[D^2]&\approx \epsilon (\sum_{i=1}^{K-l} {y_i}^2)+ 2 (\epsilon (\sum_{i=1}^{K-l} y_i))^2\approx \epsilon 
(\sum_{i=1}^{K-l} {y_i}^2)\nonumber\\
\Delta &\approx K+\epsilon l \nonumber\\
&+ \frac{\frac{\epsilon (\sum_{i=1}^{K-l} {y_i}^2)+K^2 + \epsilon l (2K+1)}{2} + (K+ \epsilon l) \epsilon 
(\sum_{i=1}^{K-l} y_i)}{K+ \epsilon l+\epsilon (\sum_{i=1}^{K-l} y_i)} \nonumber\\ 
&= \frac{3K}{2} + \epsilon l + \frac{\epsilon}{2K} (\sum_{i=1}^{K-l} {y_i}^2 + (K+1)l+ K \sum_{i=1}^{K-l} {y_i}).
\label{approx0}
\end{align}
Here $\approx$ means that the left and the right sides have a $o(\epsilon)$ difference. 
We note from~\eqref{approx0} that AoI is a function of $l$ and the indices $y_i$'s, and does not depend on $\{ c_{x_j}, j \in A \}$.

For a fixed $l$, we want to find the indices $y_i$ such that AoI is minimized. From~\eqref{approx0}, it can be seen that $y_i^* = i, i \in [K-l]$. That is, $c_i^* = 1, i \in [K-l]$. By abuse of notation, use $\Delta(l)$ to denote the average AoI if $|A|=l$. Now, we optimize over $l$.
%
%
Substituting $y_i^*$ in $\Delta$, We have $\Delta(K) \approx K(3K+1)$ and 

\vspace{-0.3cm}
\begin{small}
\begin{align*}
 &\Delta(K-1) \approx K(3K+1)-2K, \\ &\Delta(K-2) \approx K(3K+1)-3K+3.
\end{align*}
\end{small}
\vspace{-0.3cm}\\
$\bullet$ \textbf{Case $K =2 $}: it can be checked that $\Delta(K-1)= \Delta( 1) < \min (\Delta(0), \Delta(2))$. In this case, one obtains $\Delta^* = 3 + \frac{5 \epsilon}{ 2 }$\\
$\bullet$ \textbf{Case $K \geq 3$}: one can check that $\Delta(K-2) \le \min (\Delta(K-1), \Delta(K))$. We prove that  $\Delta(K-2) \le  \Delta(l), \text{for } l < K-2$.
This is equivalent to showing that
\begin{align} \label{case}
\small
&\frac{K(K-l)(K-l+1)}{2} + \frac{(K-l)(K-l+1)(2(K-l)+1)}{6}   \nonumber\\ 
 & +  (3K+1)l  \geq K(3K+1)-3K+3.
\end{align}
Defining $x=K-l$, the above is equivalent to proving
\begin{align}
\label{step1}
\frac{Kx(x+1)}{2} + \frac{x (x+1) (2x+1)}{6} + 
3 K \geq (3 K + 1 ) x + 3.
\end{align}
As $x > 2$, we have $\frac{x (x+1) (2x+1)}{6}  \geq x + 3$. To get~\eqref{step1}, it is sufficient to prove 
\begin{align*}
\frac{Kx(x+1)}{2} +3K \geq 3Kx  \iff 
x^2 -5x + 6 \geq 0,  
\end{align*}
which is true for any integer $x \geq 0$. Hence, one obtains AoI as in~\eqref{Aoi_low_snr} after simplification.
%
\end{IEEEproof}
\subsection{AoI Approximation}
In this subsection, we provide a lower and upper bound for optimal average AoI which is verified to be tight in some regimes through numerical results.

\begin{figure} 
    \centering
\hbox{\hspace{-0.45cm}    \includegraphics[scale=0.6]{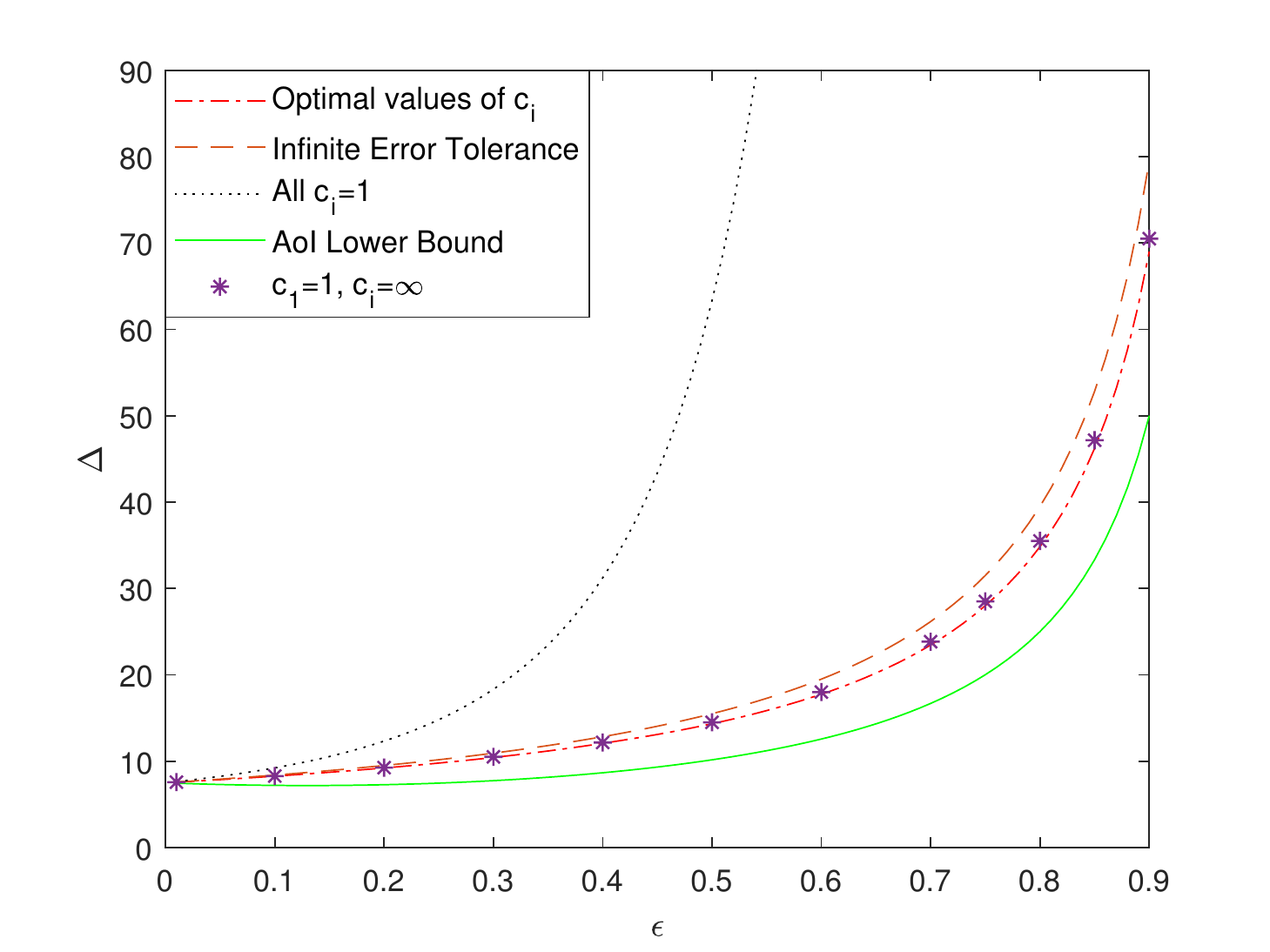} }
    \caption{Comparison of average AoI under different policies when K=5.}
    \label{comparison}
\end{figure}

\begin{lemma}\label{lem_inq}
$\Delta \geq (1+ \frac{p}{2}) (\mathbb{E}[S]+\mathbb{E}[D])$. 
\end{lemma}
\begin{proof} Using Lemmas \ref{lem1} and \ref{lem2},
\begin{align*}
\Delta &= \mathbb{E}[S]+\frac{\frac{\mathbb{E}[S^2]}{2} + \frac{\mathbb{E}[D^2]}{2}+\mathbb{E}[S]\mathbb{E}[D]}{\mathbb{E}[S]+\mathbb{E}[D]} \\&=
\mathbb{E}[S]+\mathbb{E}[D]+\frac{1}{2}\frac{\mathbb{E}[S^2] + \frac{1-p}{p} \mathbb{E}[d^2]}{\mathbb{E}[S]+\mathbb{E}[D]}
\end{align*}
Knowing that $\mathbb{E}[X^2]\geq \mathbb{E}[X]^2$, we have
\begin{align*}
 \mathbb{E}[S^2] + \frac{1-p}{p} \mathbb{E}[d^2]   \geq \mathbb{E}[S]^2 + \frac{1-p}{p} \mathbb{E}[d]^2.
\end{align*}
From Cauchy-Schwartz inequality we know that $(a+b)(c+d) \geq (\sqrt{ac}+\sqrt{bd})^2$ for positive values of $a,b,c,d$. Therefore,
\begin{align*}
    (1+ \frac{1-p}{p})(\mathbb{E}[S]^2 + \frac{1-p}{p} \mathbb{E}[d]^2) \geq  (\mathbb{E}[S]+\mathbb{E}[D])^2.
\end{align*}
Consequently we obtain
\begin{small}
\begin{align*}
    \Delta \geq \mathbb{E}[S]+\mathbb{E}[D] + \frac{p}{2} (\mathbb{E}[S]+\mathbb{E}[D]) =
    (1+ \frac{p}{2}) (\mathbb{E}[S]+\mathbb{E}[D]).
\end{align*}
\end{small}
\end{proof}
\begin{lemma} \label{rec}
$\mathbb{E}[S]+\mathbb{E}[D] \geq \frac{K}{1-\epsilon}$.
\end{lemma}
\begin{proof}
We utilize a recursive approach. 
Given a vector $\mathbf{c}=[c_1,c_2,...,c_{K}]$, we consider the policy $[c_1,c_2,...,c_{i}]$ for updates consisting of $i$ packets. For the vector $\mathbf{c}$, define $S(i), D(i), PAoI_i$ to be the successful update transmission time, the failed update transmission time, and the peak AoI, respectively, when each update consists of $i$ packets. Thus 
\begin{align} \label{recur}
&\mathbb{E}[S(i+1)] = \mathbb{E}[S(i)] + 1 + \frac{\epsilon}{1-\epsilon} - \frac{c_{i+1} \epsilon^{c_{i+1}}}{1-\epsilon^{c_{i+1}}},\\
 &\mathbb{E}[D(i+1)] =  \frac{\mathbb{E}[D(i)] + \epsilon^{c_{i+1}} (c_{i+1}+ \mathbb{E}[S(i)])}{1-\epsilon^{c_{i+1}}}. \label{recur2}
 \end{align}
We notice that
\begin{align*}
    \mathbb{E}[S(i+1)]+\mathbb{E}[D(i+1)]= \frac{1}{1-\epsilon} + \frac{\mathbb{E}[S(i)]+\mathbb{E}[D(i)]}{1-\epsilon^{c_{i+1}}}. 
\end{align*}
Defining $X_{i}=\mathbb{E}[S_i]+\mathbb{E}[D_i]$, we obtain
\begin{align} \label{key_r}
     X_{j}= \frac{j}{1-\epsilon} + \sum_{i=1}^{j-1} \frac{\epsilon^{c_{i+1}}}{1-\epsilon^{c_{i+1}}} X_{i}, \quad j \in \{2,...,K\}.
\end{align}
Clearly from equation \eqref{key_r}, one can obtain $X_K \geq \frac{K}{1-\epsilon}$ considering that $X_i \geq 0$ for $i \in \{1,...,K\}$.
\end{proof}
\begin{theorem}
The optimal value of average AoI satisfies the following bounds for any value of $\epsilon$ and $K$:
\begin{align}
\frac{K}{1-\epsilon}\left(1+\frac{(1-\epsilon)^K}{2}\right)
\le \Delta^{*} 
\le \frac{K}{1-\epsilon} \left(\frac{3}{2}+ \frac{\epsilon}{2K}\right). 
\end{align}
\end{theorem}
\begin{proof}
The upper bound is derived when choosing the policy vector $\mathbf{c}$ to be $[\infty,\infty,...,\infty]$. For the lower bound, combining resullts from Lemma \ref{lem_inq} and Lemma \ref{rec} and also knowing $p \geq (1-\epsilon)^K$ imply
\begin{align*}
    \Delta^{*} \geq (1+ \frac{p}{2}) (\mathbb{E}[S]+\mathbb{E}[D]) \geq \frac{K}{1-\epsilon}(1+\frac{(1-\epsilon)^K}{2}).
\end{align*}
\end{proof}
In Figure \ref{comparison} we compare average AoI of different policies. It can be seen that when $\epsilon$ is small the lower bound is tight. It is also clear that the policy $\mathbf{c}=[\infty,\infty, \ldots,\infty]$ outperforms the zero error toleration policy and is close to the optimal policy. Furthermore, because  $c_1^*=1$ for an optimal policy, the policy $\mathbf{c} = [ 1, \infty,\ldots,\infty]$ tightly compares to the optimal policy for average AoI.

\section{Minimizing peak age of information} \label{sec:peak}
We derive peak average age of information (PAoI) in special cases in this section. Moreover, we derive bounds on the optimal $c_i^*$ minimizing PAoI.

For both AoI and PAoI, we observe through numerical simulation that values of $c_i^*$'s are increasing as $i$ increases. Intuitively, when we are closer to receiving the update fully, the amount of tolerated erasures should increase. We prove the correctness of the intuition when minimizing PAoI. 
\begin{theorem} \label{increasing}
The optimal policy $\mathbf{c}^*$ minimizing PAoI satisfies 
${c_1}^* \leq {c_2}^* \leq ... \leq {c_K}^*$.
\end{theorem}
\begin{proof}
Consider the vector $\mathbf{c}^*=[{c_1}^*,{c_2}^*,...,{c_K}^*]$ that minimizes PAoI. Assume $j+1$ is the smallest index such that ${c_{j+1}}^*< {c_{j}}^*$. We prove that vector $\mathbf{c}^\prime=[{c_1}^*,{c_2}^*,...,{c_{j-1}}^*,{c_{j+1}}^*, {c_{j}}^*,{c_{j+2}}^*,...,{c_K}^*]$ achieves a smaller PAoI compared to $\mathbf{c}^*$, which is in contradiction to the assumption that $\mathbf{c}^*$ is optimal. 
Therefore, the assumption that we have the condition of ${c_{j+1}}^*< {c_{j}}^*$ somewhere in vector $\mathbf{c}^*$ is wrong. Consequently, for the optimal vector $\mathbf{c}^*$ we have ${c_{1}}^* \leq {c_2}^* \leq ... \leq {c_K}^*$. As we can see from the expression of $\mathbb{E}[S]$ in \eqref{expect_S}, it is indifferent of any permutation in the vector $\mathbf{c}$ and therefore remains the same for $\mathbf{c}^{\prime}$. Hence for comparing PAoI of vectors $\mathbf{c}$ and $\mathbf{c}^{\prime}$ we just have to compare $\mathbb{E}[D]$.
We have only changed the position of $2$ adjacent parameters $c_j$ and $c_{j+1}$ (not their value), 
One can show that the difference of the 2 PAoI of $\mathbf{c}$ and $\mathbf{c}^{\prime}$ is only left with $2$ terms, and PAoI of  $\mathbf{c}^{\prime}$ is smaller. 
We indicate this $2$ terms by $L_\mathbf{c}, L_{\mathbf{c}^{\prime}}$ and the goal is to show $L_{\mathbf{c}^{\prime}} < L_\mathbf{c}$.

\vspace{-0.1cm}
\begin{small}
\begin{align*}
  &L_\mathbf{c}= (1-\epsilon)^{j-1} \epsilon^{c_j} \sum_{i_1=0}^{c_1-1}\dots \sum_{i_{j-1}=0}^{c_{j-1}-1} (\sum_{r=1}^{j-1} i_r +c_j +j-1) \epsilon^{\sum_{r=1}^{j-1}}  \\
  &+(1-\epsilon)^{j} \epsilon^{c_{j+1}} \sum_{i_1=0}^{c_1-1}\dots  \sum_{i_{j-1}=0}^{c_{j-1}-1} \sum_{i_j=0}^{c_{j}-1} (\sum_{r=1}^{j-1}+i_j+ i_r +c_{j+1} +j) \epsilon^{\sum_{r=1}^{j}}\\
  &+(c_{j+1}+j) \epsilon^{c_{j+1}} (1-\epsilon^{c_1})\dots(1-\epsilon^{c_{j-1}})(1-\epsilon^{c_{j}}),\\
  &L_{\mathbf{c}^{\prime}}= (1-\epsilon)^{j-1} \epsilon^{c_{j+1}} \sum_{i_1=0}^{c_1-1}\dots \sum_{i_{j-1}=0}^{c_{j-1}-1} (\sum_{r=1}^{j-1} i_r +c_{j+1} +j-1) \epsilon^{\sum_{r=1}^{j-1}}  \\
  &+(1-\epsilon)^{j} \epsilon^{c_{j}} \sum_{i_1=0}^{c_1-1}\dots \sum_{i_{j-1}=0}^{c_{j-1}-1} \sum_{i_{j+1}=0}^{c_{j+1}-1} (\sum_{r=1}^{j-1}+i_{j+1}+ i_r +c_{j} +j) \epsilon^{\sum_{r=1}^{j}}\\
  &+(c_{j}+j) \epsilon^{c_{j}} (1-\epsilon^{c_1})\dots(1-\epsilon^{c_{j-1}})(1-\epsilon^{c_{j+1}}).
 \end{align*}
\end{small}
Let's define 
\begin{align*}
    A&=\sum_{i_1=0}^{c_1-1}\dots \sum_{i_{j-1}=0}^{c_{j-1}-1} (\sum_{r=1}^{j-1} i_r) \epsilon^{\sum_{r=1}^{j-1}},\\
    p_1&= (1-\epsilon^{c_1})\dots (1-\epsilon^{c_{j-1}}),\\
    x&= \sum_{i_{j}=0}^{c_{j}-1} i_{j} \epsilon^{i_{j}} = \frac{(c_{j}-1)\epsilon^{c_{j}+1}-c_j\epsilon^{c_j}+\epsilon}{(1-\epsilon)^2},\\
    y&= \sum_{i_{j}=0}^{c_{j+1}-1} i_{j+1} \epsilon^{i_{j+1}}=\frac{(c_{j+1}-1)\epsilon^{c_{j+1}+1}-c_{j+1}\epsilon^{c_{j+1}}+\epsilon}{(1-\epsilon)^2}.
\end{align*}
Then by expanding the summations and using $A$,

\vspace{-0.1cm}
\begin{footnotesize}
\begin{align*}
   & L_\mathbf{c}= (1-\epsilon)^{j-1} \epsilon^{c_j} A + (c_j+j-1)\epsilon^{c_j}(1-\epsilon^{c_1})\dots (1-\epsilon^{c_{j-1}})\\&+ (1-\epsilon)^{j} \epsilon^{c_{j+1}} (\frac{1-\epsilon^{c_1}}{1-\epsilon} \times \dots \times \frac{1-\epsilon^{c_{j-1}}}{1-\epsilon} (\sum_{i_{j}=0}^{c_{j}-1} (i_{j} \epsilon^{i_{j}}))+A \frac{1-\epsilon^{c_j}}{1-\epsilon})\\
   &+(c_{j+1}+j) \epsilon^{c_{j+1}} (1-\epsilon^{c_1})\dots(1-\epsilon^{c_{j-1}})(1-\epsilon^{c_{j}}),\\
   & L_{\mathbf{c}^{\prime}}= (1-\epsilon)^{j-1} \epsilon^{c_{j+1}} A + (c_{j+1}+j-1)\epsilon^{c_{j+1}}(1-\epsilon^{c_1})\dots (1-\epsilon^{c_{j-1}})\\&+ (1-\epsilon)^{j} \epsilon^{c_{j}} (\frac{1-\epsilon^{c_1}}{1-\epsilon} \times \dots \times \frac{1-\epsilon^{c_{j-1}}}{1-\epsilon} (\sum_{i_{j+1}=0}^{c_{j+1}-1} (i_{j+1} \epsilon^{i_{j+1}}))+A \frac{1-\epsilon^{c_{j+1}}}{1-\epsilon})\\
   &+(c_{j}+j) \epsilon^{c_{j}} (1-\epsilon^{c_1})\dots(1-\epsilon^{c_{j-1}})(1-\epsilon^{c_{j+1}}).
\end{align*}
\end{footnotesize}
With cancelling out mutual factors and using $[A,p_1,x,y]$ abbreviations we obtain
\begin{align*}
    &  L_\mathbf{c} - L_{\mathbf{c}^{\prime}} = (c_j+j-1) \epsilon^{c_{j}} p_1 + x(1-\epsilon)\epsilon^{c_{j+1}}p_1 \\
    & + (c_{j+1}+j) \epsilon^{c_{j+1}} (1-\epsilon^{c_j})p_1 \\
    & - (c_{j+1}+j-1) \epsilon^{c_{j+1}} p_1 - y(1-\epsilon)\epsilon^{c_j}p_1 \\
    & - (c_{j}+j) \epsilon^{c_{j}} (1-\epsilon^{c_{j+1}})p_1.
\end{align*}
By factoring out $p_1$ from each term, substituting $x$ and $y$ formula and some algebraic simplification 
\begin{align*}
    & L_\mathbf{c} - L_{\mathbf{c}^{\prime}} = p_1 (\epsilon^{c_{j+1}}-\epsilon^{c_{j}}) > 0.
\end{align*}
The last equation is because we assumed $c_{j+1}< c_{j}$ and therefore $\epsilon^{c_{j+1}}-\epsilon^{c_{j}} >0$. Thus the proof is complete.
\end{proof}
\subsection{PAoI in the Small Erasure Probability Regime ($\epsilon \to 0$)}
\begin{theorem}
 Minimum PAoI when $\epsilon \to 0$ equals to $2K+ \epsilon (2K-1)$ and optimal values of $c_i$
 that lead to minimum PAoI are $c_1^*=1$ and any $c_j^* > 1$ for $ j \in\{2,3,...,K\}$.
\end{theorem}
\begin{proof}
Using equations \eqref{approx0} and $(27)$ we derive
\begin{align}\label{p0}
    PAoI= 2\mathbb{E}[S]+\mathbb{E}[D] = 2K+2\epsilon l + \epsilon (\sum_{i=1}^{K-l} y_i).
\end{align}
Similar analysis to  Theorem \ref{epsil0} results in $l=K-1$ being the optimal $l$. Therefore $c_1=1$ and any $c_j>1$, for $j>1$, achieves the minimum PAoI, which is $2K+\epsilon(2K-1)$.
\end{proof}
\subsection{Characteristics of $c_i$ for PAoI Minimization}
Next, we study policies that minimize PAoI. We start with the case of $K=2$ and later discuss the case of general $K$. 

Recall that $c_1^* = 1$. Thus for $K=2$, we only need to determine $c_2^*$.
\begin{theorem} \label{paoi}
For $K=2$ and any $\epsilon$, the optimal value of $c_2$ for peak AoI minimization
lies within the following range,
\begin{align*}
        \ceil{\frac{1}{1-\epsilon}} \leq {c_2}^{*} \leq \ceil{\frac{1+\epsilon}{1-\epsilon}}.
\end{align*}
\end{theorem}
\begin{proof}
By substituting $K=2$ and $c_1=1$ in equation $PAoI=2\mathbb{E}[S]+\mathbb{E}[D]$ and simplification, we achieve 
\begin{align*}
PAoI &= 2 \frac{(1-\epsilon)^2}{(1-\epsilon)(1-\epsilon^{c_2})} (\sum_{i_2=0}^{c_2-1}  (i_2 +2) \epsilon^{i_2})\\
&+\frac{1}{(1-\epsilon)(1-\epsilon^{c_2})} (\epsilon + (1-\epsilon)(c_2 +1)\epsilon^{c_2} ) \\
 &= 4+ \frac{1}{(1-\epsilon)(1-\epsilon^{c_2})} (2(c_2-1)\epsilon^{c_2+1}-2c_2 \epsilon^{c_2} +2\epsilon) + \\ &\frac{1}{(1-\epsilon)(1-\epsilon^{c_2})} (\epsilon + (1-\epsilon)(c_2 +1)\epsilon^{c_2} )  \\
&= 4+ \frac{1}{1-\epsilon} \frac{3\epsilon + \epsilon^{c_2}(1-c_2) + \epsilon^{c_2+1} (c_2-3)}{1-\epsilon^{c_2}}  \\
&=4+ \frac{3\epsilon}{1-\epsilon} + \frac{1}{1-\epsilon} \frac{\epsilon^{c_2}(1-c_2 + \epsilon c_2)}{1-\epsilon^{c_2}}.
\end{align*}
As a result, we need to minimize the last term:
\begin{align}\label{mini}
 \min_{c_2} \frac{\epsilon^{c_2}(\frac{1}{1-\epsilon}-c_2)}{1-\epsilon^{c_2}}.
\end{align}
For $c_2 \ge \ceil{\frac{1}{1-\epsilon}}$,  the term in \eqref{mini} is less than or equal $0$, and for $c_2 \le \floor{\frac{1}{1-\epsilon}}$, it is greater than or equal $0$. Therefore ${c_{2}}^{*} \ge \ceil{\frac{1}{1-\epsilon}}$.
We define $c_2 = \frac{1}{1-\epsilon} + m$ for $m \ge 0$. Equation \eqref{mini} simplifies to
\begin{align*}
\min_{c_2} \frac{(1-\epsilon)\epsilon^{{\frac{1}{1-\epsilon}}+m}(\frac{1}{1-\epsilon}-\frac{1}{1-\epsilon} - m)}{1-\epsilon^{{\frac{1}{1-\epsilon}}+m}}= \max_{m} \frac{ m\epsilon^{m}}{1-\epsilon^{{\frac{1}{1-\epsilon}}+m}}.
\end{align*}
For $m \geq \frac{\epsilon}{1-\epsilon}$ we show that 
$f(m) \coloneqq \frac{ m\epsilon^{m}}{1-\epsilon^{{\frac{1}{1-\epsilon}}+m}}$
is non-increasing, so the optimal $m$ satisfies $m^* \le \frac{\epsilon}{1-\epsilon}$. For  $m \ge 0, 0 \le \epsilon \le 1$, we know that  $\frac{ 1-\epsilon^{{\frac{1}{1-\epsilon}}+m}}{1-\epsilon^{{\frac{1}{1-\epsilon}}+m+1}} \leq 1$. When $m \geq \frac{\epsilon}{1-\epsilon}$,
\begin{align*}
\small
    &\frac{m+1}{m} \leq \frac{1}{\epsilon} \implies 
    &\frac{ (m+1)\epsilon^{m+1}}{1-\epsilon^{{\frac{1}{1-\epsilon}}+m+1}} \leq \frac{ m\epsilon^{m}}{1-\epsilon^{{\frac{1}{1-\epsilon}}+m}} .
\end{align*}
Therefore, for $m \geq \frac{\epsilon}{1-\epsilon}$, $f(m)$ is non-increasing.
Consequently, the opitimal $c_2$ satisfies ${c_2}^{*} \leq \ceil{\frac{1+\epsilon}{1-\epsilon}}$.
\end{proof}
Finding exact values of $c_i^*$'s that minimize PAoI is hard for general $K$. However, the bound in Theorem~\ref{paoi} inspires the following general bound for $c_i^*$ and consequently an upper bound for PAoI.

\begin{theorem} \label{lower_8}
For any $K$ and $\epsilon$, the value of $c_i^*$ that minimizes the PAoI satisfies
\begin{align*}
  c_1^*=1 \text{ and }   c_i^* \geq \ceil{\frac{i-1}{1-\epsilon}}, \quad i \in [2,3,...,K].
\end{align*}
\end{theorem}
\begin{proof}
Using Equations \eqref{recur} and \eqref{recur2} and following the same notation in these equations, we obtain 
\begin{align*} 
\small
PAoI_{i+1}  &= 2\mathbb{E}[S(i+1)] + \mathbb{E}[D(i+1)]\\
&=2\mathbb{E}[S(i)] + \mathbb{E}[D(i)]\\ & + 2+ \frac{2\epsilon}{1-\epsilon}
+ \frac{\epsilon^{c_{i+1}} (\mathbb{E}[S(i)]+\mathbb{E}[D(i)]-c_{i+1})}{1-\epsilon^{c_{i+1}}}\\
&= PAoI_{i}   + 2 + \frac{2\epsilon}{1-\epsilon}\\
& + \frac{\epsilon^{c_{i+1}} (\mathbb{E}[S(i)]+\mathbb{E}[D(i)]-c_{i+1})}{1-\epsilon^{c_{i+1}}}\\
    PAoI_{K} &= 2K + \frac{(2K-1)\epsilon}{1-\epsilon} \\
    &+ \sum_{j=1}^{K-1} \frac{\epsilon^{c_{j+1}} (\mathbb{E}[S_j]+\mathbb{E}[D_j]-c_{j+1})}{1-\epsilon^{c_{j+1}}}.
\end{align*}
To show that $c_i^* \geq \ceil{\frac{i-1}{1-\epsilon}}$, assume that in the optimal vector $\mathbf{c}^{*}=[c_1^*,\dots,c_K^*]$, there exists one coordinate $c_j^* \le \floor{\frac{j-1}{1-\epsilon}}$. We show that vector $\mathbf{c}^{\prime} = [c_1^*,c_2^*,...,c^{\prime}_j,c_{j+1}^*,...,c_K^*]$ for $c^{\prime}_j=\ceil{\frac{j-1}{1-\epsilon}}$ results in a smaller PAoI, which is in contradictory to the assumption that $\mathbf{c}$ is the optimal vector. 
To this end, denote by $PAoI_\mathbf{c}, PAoI_{\mathbf{c}^{\prime}}$ for PAoI resulting from vectors $\mathbf{c}, \mathbf{c}^{\prime}$, respectively. Define $X'_i=\mathbb{E}[S_i]+\mathbb{E}[D_i]$ corresponding to the vector $\mathbf{c}'$. Our goal is to prove $PAoI_{\mathbf{c}^{\prime}} <PAoI_\mathbf{c}$.
From \eqref{key_r}, it is obvious that for a given vector $\mathbf{c}$ if we change only $c_{j+1}$ we will have $X'_i = X_i, i=1,2,\dots,j$.
From \eqref{key_r}
it is clear that if we increase $c_{j+1}$ and keep the rest of $c_i$'s the same, $X'_ {i} < X_i, i=j+1,\dots,K$.
\begin{align*}
    PAoI_\mathbf{c} &= 2K+ \frac{(2K-1)\epsilon}{1-\epsilon} +  \sum_{i=1}^{j-1} \frac{\epsilon^{c_{i+1}} (X_i-c_{i+1})}{1-\epsilon^{c_{i+1}}} \\ &+\frac{\epsilon^{c_{j+1}}(X_j-c_{j+1})}{1-\epsilon^{c_{j+1}}}+ \sum_{i=j+1}^{K-1} \frac{\epsilon^{c_{i+1}} (X_i-c_{i+1})}{1-\epsilon^{c_{i+1}}},\\
    PAoI_{\mathbf{c}^{\prime}}&=2K+ \frac{(2K-1)\epsilon}{1-\epsilon} +  \sum_{i=1}^{j-1} \frac{\epsilon^{c_{i+1}} (X_i-c_{i+1})}{1-\epsilon^{c_{i+1}}} \\ &+\frac{\epsilon^{{c^{\prime}}_{j+1}}(X_j-{c^{\prime}}_{j+1})}{1-\epsilon^{{c^{\prime}}_{j+1}}}+ \sum_{i=j+1}^{K-1} \frac{\epsilon^{c_{i+1}} (X^{{\prime}}_i-c_{i+1})}{1-\epsilon^{c_{i+1}}}. 
\end{align*}
The first $3$ terms are exactly the same and for the rest we show
\begin{align}
    &\frac{\epsilon^{c_{i+1}} (X_i-c_{i+1})}{1-\epsilon^{c_{i+1}}} > 
    \frac{\epsilon^{c_{i+1}} (X^{{\prime}}_i-c_{i+1})}{1-\epsilon^{c_{i+1}}}, j+1 \le i \le K-1
   , \label{lower_bound1}\\
     &\frac{\epsilon^{c_{j+1}}(X_j-c_{j+1})}{1-\epsilon^{c_{j+1}}} > \frac{\epsilon^{{c^{\prime}}_{j+1}}(X_j-{c^{\prime}}_{j+1})}{1-\epsilon^{{c^{\prime}}_{j+1}}}.
     \label{lower_bound2}
    \end{align}
Inequality \eqref{lower_bound1} holds since $X^{\prime}_{i} < X_{i}$ for $j+1 \le i \le K-1$. 
Inequality \eqref{lower_bound2} holds by the following
\begin{align*}
   X_j-c_{j+1} &\geq X_j- \floor{\frac{j}{1-\epsilon}} \geq 0 ,\\ 
   X_j-c_{j+1} &> X_j-c^{\prime}_{j+1} 
   \\
   \frac{\epsilon^{c_{j+1}}}{1-\epsilon^{c_{j+1}}} &> \frac{\epsilon^{c^{\prime}_{j+1}}}{1-\epsilon^{c^{\prime}_{j+1}}} >0.
    \end{align*}
    \end{proof}
\begin{figure} 
    \centering
\hbox{\hspace{-0.2cm}    \includegraphics[scale=0.6]{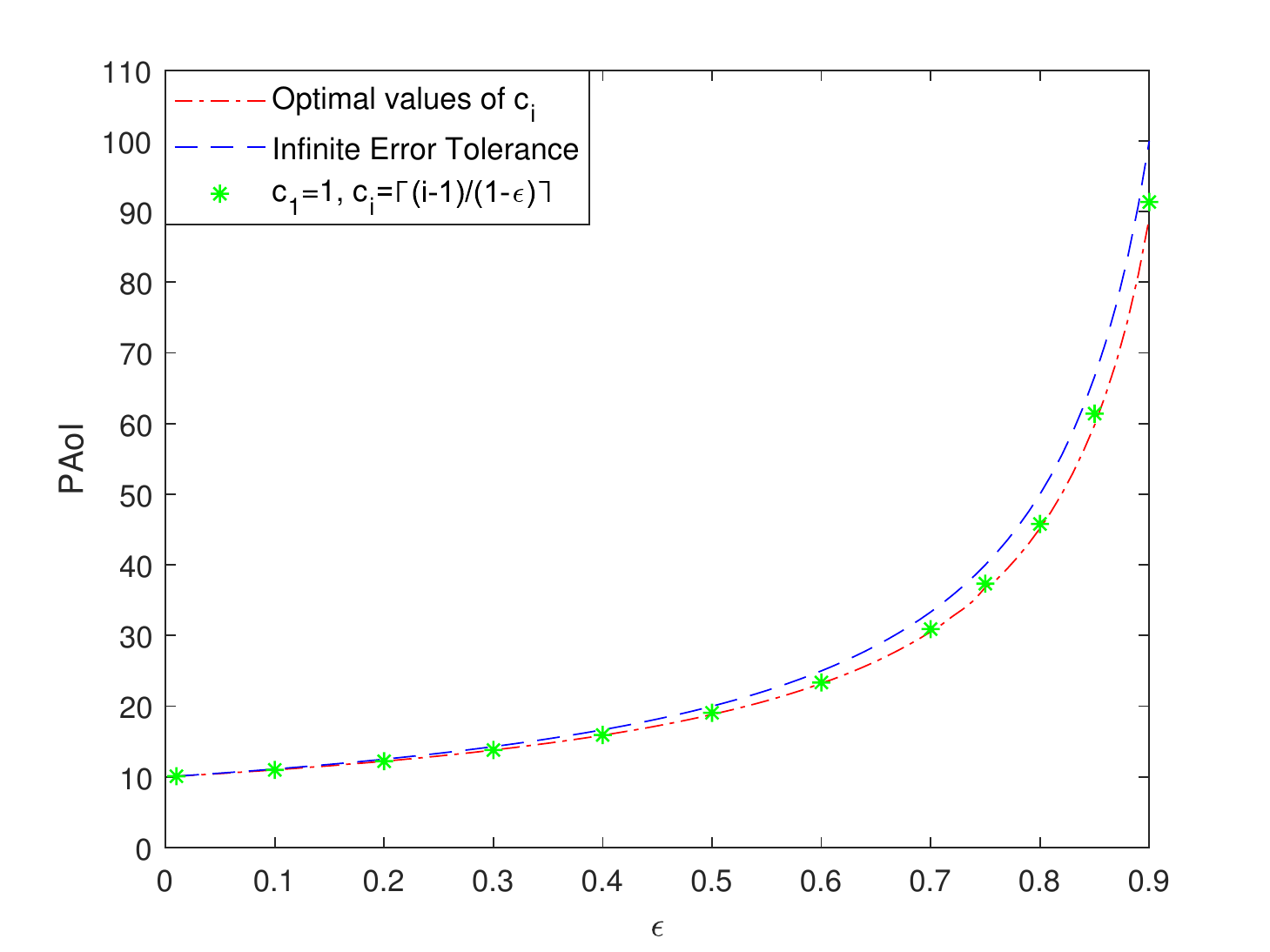} }
    \caption{Comparison of PAoI under different policies when K=5.}
    \label{comparison_paoi}
\end{figure}
In Figure \ref{comparison_paoi}, we observe that setting $c_i$'s as the lower bound from Theorem \ref{lower_8} outperforms infinite error tolerance policy and its PAoI difference from the  optimal policy is negligible.
\section{Conclusion} \label{sec:conclude}
In this paper, we studied the optimal error toleration policy for AoI minimization during transmission of an update in an erasure channel with feedback. We obtained optimal policies for special senarios and bounds of AoI for the general case.

\bibliography{reference}
\bibliographystyle{IEEEtran}

\end{document}